\documentclass[aps,pra,reprint,longbibliography,floatfix]{revtex4-1}
\usepackage{etex}
\usepackage{graphicx}
\usepackage{caption}
\usepackage{subcaption}
\usepackage{mathrsfs}
\usepackage{amsfonts}
\usepackage{dsfont}
\usepackage{times}
\usepackage{amsmath}
\usepackage{amsthm}
\usepackage{leftidx}
\usepackage{tikz}
\usepackage{tikz-network}
\usepackage{color}
\usepackage[ bookmarks=true, colorlinks, linkcolor=blue, urlcolor=blue, citecolor=blue, plainpages=false, pdfpagelabels, final, breaklinks=true ]{hyperref}

\newcommand{\Tr}{\operatorname{Tr}}
\usepackage{mathtools}
\usepackage{array}
\usepackage{multirow}
\usepackage{bbold}
\usepackage{kotex}
\usepackage{ragged2e}%

\newcommand{\ket}[1]{|#1\rangle}
\newcommand{\bra}[1]{\langle#1|}

\newcolumntype{L}[1]{>{\raggedright\let\newline\\\arraybackslash\hspace{0pt}}m{#1}}
\newcolumntype{C}[1]{>{\centering\let\newline\\\arraybackslash\hspace{0pt}}m{#1}}
\newcolumntype{R}[1]{>{\raggedleft\let\newline\\\arraybackslash\hspace{0pt}}m{#1}}

\newtheorem{corollary}{Corollary}
\newtheorem{theorem}{Theorem}

\theoremstyle{definition}

\DeclareRobustCommand{\orcidicon}{%
	\begin{tikzpicture}
	\draw[lime, fill=lime] (0,0) 
	circle [radius=0.16] 
	node[white] {{\fontfamily{qag}\selectfont \tiny ID}};
	\draw[white, fill=white] (-0.0625,0.095) 
	circle [radius=0.007];
	\end{tikzpicture}
	\hspace{-2mm}
}
\foreach \x in {A, ..., Z}{%
	\expandafter\xdef\csname orcid\x\endcsname{\noexpand\href{https://orcid.org/\csname orcidauthor\x\endcsname}{\noexpand\orcidicon}}
}

\begin{document}
\title{Port-based entanglement teleportation via noisy resource states}
\author{Ha Eum Kim\orcidA{}}
\email{hekim007@korea.ac.kr}
\affiliation{Department of Physics, Korea University, Seoul 02841, Korea}
\author{Kabgyun Jeong\orcidB{}}
\email{kgjeong6@snu.ac.kr}
\affiliation{Research Institute of Mathematics, Seoul National University, Seoul 08826, Korea}
\affiliation{School of Computational Sciences, Korea Institute for Advanced Study, Seoul 02455, Korea}

\date{\today}

\begin{abstract}
Port-based teleportation (PBT) represents a variation of the standard quantum teleportation and is currently being employed and explored within the field of quantum information processing owing to its various applications.
In this study, we focus on PBT protocol when the resource state is disrupted by local Pauli noises.
Here, we fully characterise the channel of the noisy PBT protocol using Krauss representation.
Especially, by exploiting the application of PBT for entanglement distribution necessary in realizing quantum networks, we investigate entanglement transmission through this protocol for each qubit considering noisy resource states, denoted as port-based entanglement teleportation (PBET).
Finally,
we derive upper and lower bounds 
for the teleported entanglement as a function of the initial entanglement and the noises.
Our study demonstrates that quantum entanglement can be efficiently distributed by protocols utilizing large-sized resource states in the presence of noise and is expected to serve as a reliable guide for developing optimized PBET protocols.
To obtain these results, 
we address that the order of entanglement of two qubit states is preserved through the local Pauli channel, and identify the boundaries of entanglement loss through this teleportation channel.
\end{abstract}
\maketitle

\section{Introduction}

Quantum teleportation, initially proposed by Bennett \textit{et al}. \cite{bennett1993teleporting}, is an innovative protocol for transmitting an unknown quantum state to a spatially separated receiver without the need to physically transmit qubits.
What makes this classically unexpected phenomenon possible is that the sender and the receiver previously shared a maximally entangled state, creating a long range quantum correlation. To implement the long-distance entanglement, which is the core of the quantum networks\cite{kalb2017entanglement,daiss2021quantum,pompili2021realization}, entanglement swapping \cite{pan1998experimental} and entanglement teleportation \cite{lee2000entanglement} have been studied theoretically, and have been experimentally developed using optical fibers and satellites \cite{Valivarthi2016quantum,Ren2017ground,Barasinski2019demonstration}.

The concept of teleportation has also become an essential part of quantum information theory. It has been experimentally demonstrated in various systems \cite{Barrett2004Deterministic,olmschenk2009quantum} and studied theoretically from diverse perspectives \cite{Lee2021Quantum,zhou2000Methodology,chitambar2023duality}. Additionally, various application protocols have been proposed to enhance performance or meet specific purposes, such as bidirectional quantum teleportation \cite{ikken2023bidirectional}, controlled quantum teleportation \cite{dong2011controlled,Jeong2016minimal}, continuous variable teleportation \cite{braunstein1998teleportation}, and their combinations \cite{zha2013bidirectional,kirdi2023controlled} .

One such variant, port-based teleportation (PBT) \cite{ishizaka2008asymptotic,ishizaka2009quantum,Wang2016port,studzinski2017port,Christandl2021asymtotic,Mozrzymas2018optimal,Strelchuk2023minimal,jeong2020generalization,studzinski2022efficient} has been introduced. While the standard teleportation protocol requires recovering operation by the receiver at the end, PBT does not require this quantum process. Instead, the receiver simply selects a port according to the classical information related to the sender's measurement outcome.

The generality of the PBT protocol enables a variety of applications in fields such as cryptography \cite{Beigi2011simplified}, holography \cite{may2022complexity}, quantum computing \cite{sedlak2019optimal,quintino2019reversing}, and quantum telecloning \cite{karlsson1998quantum,murao1999quantum}. Furthermore, it offers insights into non-local measurements of multi-partite states and contributes to studies on communication complexity \cite{buhrman2016quantum} and quantum channels \cite{Pirandola2019fundamental}.

Despite its prospective value, this protocol has not been experimentally realized due to the challenging implementation of the joint positive operator-valued measure (POVM). According to a recent study \cite{grinko2023efficient}, PBT was expressed as an efficient quantum circuit for the first time using the symmetry existing within the joint POVM. This opens the way for experimental realizations of PBT on a variety of physical platforms in the near future.

However, in a realistic implementation of the protocol, noise is inevitable due to the unavoidable interaction with the environment while the entangled state is distributed to the two parties.
Consequently, the relation between quantum teleportation and noise on resource state has been investigated \cite{oh2002fidelity,fonseca2019high,knoll2014noisy}.
As a fundamental finding, Popescu \textit{et al}. \cite{popescu1994bell} proved entanglement of resource state must remain to have higher fidelity than a classical communication process.
As an initial study of PBT with noisy resource states, 
Pereira \textit{et al}. \cite{Pereira2021characterising} investigated the amplitude damping channel by means of PBT with a finite number of ports.

In this context, we explore PBT with resource states affected by Pauli noise, regardless of the number of ports. Pauli noise is one of the quantum noise model capable of describing various types of incoherent noise including dephasing, bit-flip, and depolarizing. 
Moreover, the randomized compilation technique \cite{wallman2016noise,hashim2021randomized} has recently been developed, enabling general noise models to be tailored to stochastic Pauli errors.
This technique works by applying independently randomized gates into a logical circuit in such a way that leaves the effective logic circuit unchanged.
From the same perspective, we consider the maximally entangled states rotated by independently randomized local operations and the corresponding measurement matching the direction of each rotated state.
While leaving the PBT protocol unaffected, this transformation tailors the assorted general noises arising in the resource state into stochastic Pauli noise.

In this work, we derive and represent the analytic expression of PBT using resource states affected by local Pauli noises, employing the Krauss operator formalism. We find that the channel can be decomposed into a chain of channels of number of ports and environmental noises. Additionally, we consider the scheme to teleport an entangled states with PBT for each qubit, referred to as port-based entanglement teleportation (PBET), in the presence of noise. More precisely, we investigate the upper and lower bounds of transmitted entanglement.

This paper is organized as follows.
In section \ref{sec:sub:fids}, we first revisit the definitions and relationships between teleportation fidelity and entanglement fidelity. Additionally, we provide an overview of the positive partial transposition (PPT) criterion \cite{peres1996separability,horodecki2009quantum}, which is one of the quantities that satisfies the condition for the measure of entanglement \cite{vedral1997quantifying}.
We end this section with a brief introduction of the PBT protocol in section \ref{sec:sub:intro_PBT}.
Moving on to section \ref{sec:pl_ch}, we prove the preservation of entanglement order (see section \ref{sec:sub:ent_ord_pre}) and investigate boundaries of entanglement (see section \ref{sec:sub:bnd_ent}) for two-qubit states affected by Pauli channels. 
These channels not only represent the noise we encountered but also serve as the channel through which teleportation is accomplished, as explained in the subsequent section.
In section \ref{sec:sub:ch_descpt}, we describe the noisy PBT protocol as a chain of channels with number of ports and environment noise.
In section \ref{sec:sub:ent_of_pbt}, we present the upper and lower bounds of entanglement teleportation.
Finally, discussions and remarks are provided in section \ref{sec:con}.

\section{Preliminary}
\label{sec:preliminary}
\subsection{Teleportation fidelity, entanglement fidelity, measure of entanglement}
\label{sec:sub:fids}

We start with the basic definitions of teleportation fidelity, entanglement fidelity, and PPT criterion. 
The fidelity of communication over a teleportation channel $\Lambda$ is  given by
    \begin{align}
        f(\Lambda):=\int\bra{\psi}\Lambda
            \bigl(\ket{\psi}\bra{\psi}\bigr)
            \ket{\psi}
            d\psi,
    \end{align}
 where the integral is performed with respect to the uniform distribution $d\psi$ over all $d$-dimensional pure states \cite{horodecki1999general}. The entanglement fidelity of channel $\Lambda$ is given by
    \begin{align}
    \label{def:ent_fid}
        F(\Lambda):=\Tr
        \Bigl[
        \hat{\Psi}^+
        \bigl[
        (\Lambda\otimes \mathbb{I})\hat{\Psi}^+
        \bigr]
        \Bigr],
    \end{align} 
where $\hat{\Psi}^+:=\frac{1}{d}\sum^{d-1}_{i,j=0}\ket{i}\bra{j}\otimes\ket{i}\bra{j}$ is the density  matrix of the bipartite maximally entangled state with Schmidt rank $d$,
and $\mathbb{I}$ is an identity channel.
Furthermore, a universal relation \cite{horodecki1999general} between teleportation fidelity and entanglement fidelity is expressed by
    \begin{align}
    \label{eq:rel_tel_fid_and_ent_fid}
        f(\Lambda)=\frac{2F(\Lambda)+1}{3}.
    \end{align}
Given that classical communication can have a maximum fidelity of $f(\Lambda)=2/(d+1)$, teleportation with a fidelity exceeding this classical limit holds significant importance \cite{badziag2000local}.

In this study, the PPT criterion \cite{peres1996separability} was considered as the measure of entanglement. Let $\hat{\rho}$ be density matrix in a multi-qudit system and $\hat{\rho}^{T_1}$ be the partial transpose of the first qudit of $\hat{\rho}$. Then if $\hat{\rho}$ is separable, all the eigenvalues of $\hat{\rho}^{T_1}$ are guaranteed to be positive.
In particular, for a two-qubit systems, this statement becomes not only a necessary condition to be separable, but also a sufficient condition \cite{HORODECKI1996separability}.
Therefore, the measure of entanglement for a two-qubit system can be defined by the negative eigenvalue $\lambda^-$ of matrix $\hat{\rho}^{T_1}$, if it exists. This can be expressed in the form of
    \begin{align}
    \label{eq:measure_of_entanglement}
        \mathcal{M}(\hat{\rho}):=\max
        \left[0,-2\lambda^-
        \right],
    \end{align}
where the matrix cannot have more than one negative eigenvalue in this case.

\subsection{Port-based teleportation}
\label{sec:sub:intro_PBT}
In this paper, we consider the standard PBT scheme first proposed in \cite{ishizaka2008asymptotic}, and restrict our discussion to qubit systems with $d=2$.
To start with, let's assume that the sender, say Alice, and the receiver, say Bob, share a 2$N$-qudit pure state described by
    \begin{align}
    \label{eq:resource_sts}
        \hat{\Phi}_{\vec{B}\vec{A}}:=
        \hat{\Psi}^+_{B_1A_1}\otimes\cdots\otimes
        \hat{\Psi}^+_{B_NA_N}
    \end{align}
where $\hat{\Psi}^+:=\ket{\Psi^0}\bra{\Psi^0}$ is the density matrix of Bell state defined as
    \begin{align}
    \label{eq:bell_sts}
        \ket{\Psi^0}:=\frac{\ket{0}\otimes\ket{0}
        +\ket{1}\otimes\ket{1}}
        {\sqrt{2}}.
    \end{align}
Let $A_i$ and $B_i$ denote the $i\in\{1,2,...,N\}$th port of Alice and Bob's resource systems, given by $\vec{A}:=\{A_1,A_2,...,A_N\}$ and $\vec{B}:=\{B_1,B_2,...,B_N\}$, respectively. In this subsection, we use the subscript of a state or an matrix to indicate where it operates.

the sender, say Alice, jointly measures qudit containing a unknown state she wants to transmit along with one part of the pre-shared entangled state, and transmits the measurement outcome to the receiver, say Bob, using classical communication channel. 

Alice prepares her qubit $C$ with an unknown state $\hat{\rho}$.
To transmit the state to Bob, Alice jointly measures qubits $\vec{A}$ and $C$ with measurement described by POVM whose elements $\left\{\hat{\Pi}^{(i)}\right\}^N_{i=1}$ are defined as
    \begin{align}
    \label{eq:povm_meas}
        \hat{\Pi}^{(i)}:=\hat{\Pi}'^{(i)}+\hat{\Delta},\;\;\;
        \mathrm{with}\;
        \hat{\Delta}:=\frac{1}{N}\sum_{j=1}^{N}\hat{\Pi}'^{(j)}
    \end{align}
in the form of the square-root measurement (SRM),
    \begin{align}
        \hat{\Pi}'^{(i)}:=
        \hat{T}^{-\frac{1}{2}}
        \hat{\tau}^{(i)}
        \hat{T}^{-\frac{1}{2}}\;\;;\;\;
        \hat{T}:=\sum^{N}_{j=1}\hat{\tau}^{(j)}.
    \nonumber
    \end{align}
$\hat{\tau}^{(i)}$ are called \textit{signal} states, expressed as
    \begin{align}
        \hat{\tau}^{(i)}_{C\vec{A}}
        :=\frac{1}{2^{N-1}}\hat{\Psi}^+_{CA_i}\otimes\hat{I}_{\bar{A}_i},
    \nonumber
    \end{align}
where $\hat{I}_{\bar{A}_i}$ represents the identity operator acting on the rest of qubits excluding $A_i$ expressed as
$\bar{A}_i:=\vec{A}\backslash\{A_i\}$.
After the measurement, Alice sends her outcome $i$ using a classical communication channel. Then Bob selects and retains qubit $B_i$, also referred to as a port, corresponding to the outcome $i$, and removes the rest of his qubits. The channel of PBT is expressed as
    \begin{align}
    \label{eq:PBT_channel}
        \Lambda(\hat{\rho}_C)
        &=\sum_{j=1}^N
        \Tr_{\vec{A}\bar{B}_jC}
        \left[
        \left(
        \hat{I}_{\vec{B}}\otimes
        \hat{\Pi}_{\vec{A}C}^{(j)}
        \right)
        \left(
            \hat{\Phi}_{\vec{B}\vec{A}}
            \otimes\hat{\rho}_{C}
        \right)
        \right],
    \end{align}
where $\bar{B}_i:=\vec{B}\backslash\{B_i\}$.
According to Eq. (\ref{def:ent_fid}), the exact form of entanglement fidelity $F(\Lambda)$ is given by \cite{ishizaka2008asymptotic,ishizaka2009quantum}
    \begin{align}
    \label{eq:ent_fid_PBT}
        \frac{1}{2^{N+3}}\sum^N_{k=0}
        \left(
        \frac{N-2k-1}{\sqrt{k+1}}+\frac{N-2k+1}{\sqrt{N-k+1}}
        \right)^2
        {N \choose k}
        .
    \end{align}

For asymptotic limit of $N\rightarrow \infty$, the entanglement fidelity approximates to
    \begin{align}
        F(\Lambda)\rightarrow1-\frac{3}{4N},
    \end{align}
while the teleportation fidelity approximates to
    \begin{align}
        f(\Lambda)\rightarrow1-\frac{1}{2N}
    \end{align}
according to Eq. (\ref{eq:rel_tel_fid_and_ent_fid}).

\section{Pauli channel}
\label{sec:pl_ch}

The interaction between quantum systems and their surrounding environments introduces noise into the system's state, causing a degradation of coherence and purity.
To mathematically describe the noise, one widely adopted approach is the Kraus operator formalism \cite{nielsen2010quantum}.
Specifically, consider a state $\hat\rho$ and a trace-preserving quantum operation $\mathcal{E}$. The action of the quantum operation on the state can be expressed as 
\begin{align}
\mathcal{E}(\hat\rho)=\sum_k\hat{E}_k\hat\rho\hat{E}_k^\dagger,
\nonumber
\end{align}
where $\{\hat{E}_k\}$ denotes the set of Kraus operators of $\mathcal{E}$ satisfying the completeness relation $\sum_k\hat{E}_k^\dagger\hat{E}_k=\hat{I}$.
This method offers the advantage of describing states without requiring a deep physical understanding of the interaction of the system we are interested in with its environment.

In this study, we focused on thePauli noise to represent the noise generated by the resource state interacting with the local environment.
This encompasses incoherence noise, including bit flip, bit-phase flip, phase flip, and depolarizing noise.
Table \ref{tb:pl_op} provides the Kraus operator representation and notation; the representation for general Pauli noise at a single qubit is
\begin{align}
    \mathcal{E}_{\vec{p}}(\hat\rho)
    =\sum_{k=0}^3\hat{E}_k\hat\rho\hat{E}_k^\dagger
    \;\;;\;\;
    \hat{E}_i
    :=
    \sqrt{\frac{p_i}{4}}\hat{\sigma}_i,
    \;
    p_i>0,
\label{eq:quan_op_pl}
\end{align}
where $\hat{\sigma}_0$ represents the identity matrix, $\hat{\sigma}_i$ for $i=1,2,3$ are Pauli matrices, and the channel probabilities $\vec{p}:=\{p_1,p_2,p_3\}$ satisfy the relation $\sum_{k=0}^3p_k=4$.
Let us define one of the quantities we use to describe Pauli noise, namely the average of channel probabilities of 
$\mathcal{E}_{\vec{p}}$ as
    \begin{align}
    \label{eq:avg_of_p}
        \Omega
        :=
        \frac{1}{3}\sum_{j=1}^3p_j.
    \end{align}

    \begin{table}
    \begin{center}
    \begin{tabular}{ L{6em} | C{4em} | L{15em} }
    \hline
    \centering{Noise} & Notation & Kraus operator\\[1.2ex] 
    \hline
    \hline
    Bit flip &
    $\mathcal{E}^1_p$ &
    $\hat{E}_0=\sqrt{1-\frac{3p}{4}}\hat{\sigma}_0\;,\;
    \hat{E}_1=\sqrt{\frac{3p}{4}}\hat{\sigma}_1$\\[1.2ex]
    \hline
    Bit-phase flip & 
    $\mathcal{E}^2_p$ &
    $\hat{E}_0=\sqrt{1-\frac{3p}{4}}\hat{\sigma}_0\;,\;
    \hat{E}_1=\sqrt{\frac{3p}{4}}\hat{\sigma}_2$\\[1.2ex]
    \hline    
    Phase flip &
    $\mathcal{E}^3_p$ & 
    $\hat{E}_0=\sqrt{1-\frac{3p}{4}}\hat{\sigma}_0\;,\;
    \hat{E}_1=\sqrt{\frac{3p}{4}}\hat{\sigma}_3$\\[1.2ex]
    \hline    
    \multirow{2}{4em}{Depolarizing}
    & \multirow{2}{2em}{$\mathcal{E}^{\mathrm{dep}}_p$}
    &$\hat{E}_0=\sqrt{1-\frac{3p}{4}}\hat{\sigma}_0\;,\;
    \hat{E}_1=\sqrt{\frac{p}{4}}\hat{\sigma}_1$\;,
    \\
    &
    &$\hat{E}_2=\sqrt{\frac{p}{4}}\hat{\sigma}_2\;,\;
    \hat{E}_3=\sqrt{\frac{p}{4}}\hat{\sigma}_3$\\[1.2ex]
    \hline
    \hline
    \end{tabular}
    \end{center}
    \caption{Kraus representation of one of the most common types of noise with the form of Pauli noise. The noise values in the table are aligned to have the same average of channel probabilities.
    \label{tb:pl_op}}
    \end{table}

The Kraus representation of Pauli noise, not only explains the noise but also facilitates the description of noisy PBT scheme. Therefore, we refer to it as the Pauli channel, which will be further explored in section \ref{sec:noisy_pbt}, demonstrating that the properties observed in the PBT protocol is attributable to this channel.

Within the remainder of this section, we first prove that the entanglement order of any two states remains invariant under local Pauli channels (section \ref{sec:sub:ent_ord_pre}). Subsequently, in section \ref{sec:sub:bnd_ent}, we investigate upper and lower bounds regarding the reduction of entanglement caused by the Pauli channel.

\subsection{Preservation of entanglement order}
\label{sec:sub:ent_ord_pre}
In this section, we demonstrate that when the entanglement of one arbitrary state surpasses that of another state, the entanglement of the former remains greater even after going through a local Pauli channel.
We begin by establishing the preservation of entanglement order for a single qubit Pauli channel. The statement is as follows
    \begin{theorem}\label{th:ord_pre_pl}
    If the entanglement of bipartite mixed states
    $\hat{\rho}$
    and 
    $\hat{\tau}$
    satisfies
    \begin{align}
    \label{eq:thm:ent_cond}
    \mathcal{M}\left(
            \hat{\rho}
    \right)
    >
    \mathcal{M}\left(
            \hat{\tau}
    \right)>0,
    \end{align}
    then it follows that
        \begin{align}
            \mathcal{M}
            \Bigr[
            (\mathcal{E}_{\vec{p}}\otimes\mathbb{I})
                \hat{\rho}
            \Bigr]
            >
            \mathcal{M}
            \Bigr[
            (\mathcal{E}_{\vec{p}}\otimes\mathbb{I})
                \hat{\tau}
            \Bigr]
        \end{align}
    for all the channel probabilities $\vec{p}$ that satisfies
        \begin{align}
        \label{eq:thm:cond_pos_ent}
            \mathcal{M}
            \Bigr[
            (\mathcal{E}_{\vec{p}}\otimes\mathbb{I})
                \hat{\rho}
            \Bigr]
            >
        0,
        \end{align}
    where $\mathbb{I}$ is an identity channel.
\end{theorem}
\begin{proof}
    The partial transpose of $\hat{\rho}$ at second qubit can be written as
    \begin{align}
        \hat{\rho}^{T_2}
        =\frac{1}{4}+\hat{R},
    \nonumber
    \end{align}
    where $\hat{R}
    :=
    \sum_{i=1}^{3}r_{i,0}\hat{\sigma}_i
    \otimes\hat{\sigma}_0
    +\hat{S}+\sum_{i,j=1}^{3}
    r_{i,j}\hat{\sigma}_i
    \otimes\hat{\sigma}_j$ and $\hat{S}:=\sum_{i=1}^{3}r_{0,i}
    \hat{\sigma}_0
    \otimes\hat{\sigma}_i$.
    The first part of the proof is to show that the theorem holds for a small magnitude of probabilities.
    We provide a specific example for depolarized channels because can be easily generalized to all channels.
    Suppose a depolarizing channel $\mathcal{E}^{\mathrm{dep}}_{\epsilon}$ with $\epsilon \ll 1$.
    Then the partial transpose of $\hat{\rho}$ after depolarizing the channel at second qubit becomes
    \begin{align}
        \Bigr[\Bigl(\mathcal{E}^{\mathrm{dep}}_{\epsilon}
                \otimes\mathbb{I}\Bigl)
        \hat{\rho}
        \Bigl]^{T_2}
        &=\frac{1}{4}+(1-\epsilon)
        \hat{R}
            +\epsilon \hat{S}.
    \nonumber
    \end{align}
    Given that partial transpose of the mixed state has only single negative eigenvalue for two-qubit system \cite{HORODECKI1996separability}, the gap between other eigenvalues is finite at the region when entanglement is larger than zero. This means that the degenerate case only takes place when all the eigenvalues are positive, which is not the region we are interested in. This is the reason why Eq. (\ref{eq:thm:cond_pos_ent}) is needed as the condition.
    Thus, we can always find a small value of $\epsilon$ for which the first order perturbation theory works \cite{lowdin1951note}.
    Let $\ket{G_{\rho}}$ be the eigenvector of the negative eigenvalue that is one and half that of the entanglement. It satisfies
    \begin{align}
        \hat{\rho}
        ^{T_2}\ket{G_{\rho}}
        =-\frac{1}{2}\mathcal{M}(\hat{\rho})\ket{G_{\rho}}.
        \nonumber
    \end{align}
    Then the entanglement of the state through the depolarizing channel is approximated by
    \begin{align}
        \mathcal{M}\Bigl[\Bigl(\mathcal{E}^{\mathrm{dep}}_{p}
                \otimes\mathbb{I}\Bigl)\hat{\rho}\Bigr]
        &\approx\mathcal{M}(\hat{\rho})
            -
            \epsilon
            (\mathcal{M}(\hat{\rho})+2s)
            ,
        \nonumber
    \end{align}
    where $s_{\rho}
    :=\bra{G_{\rho}}
    \hat{S}
    \ket{G_{\rho}}$.
    
    Subsequently, the difference in entanglement between 
    $\hat{\rho}$
    and 
    $\hat{\tau}$
    through the depolarizing channel is given by
    \begin{align}
        \mathcal{M}\Bigl[\Bigl(\mathcal{E}^{\mathrm{dep}}_{p}
        \otimes\mathbb{I}\Bigl)
        \hat{\rho}
        \Bigr]
        -
        \mathcal{M}\Bigl[\Bigl(\mathcal{E}^{\mathrm{dep}}_{p}
        \otimes\mathbb{I}\Bigl)
        \hat{\tau}
        \Bigr]
        \approx
        dM-
            \epsilon
            (dM+2ds)
        ,
        \nonumber
    \end{align}
    where $dM:=\mathcal{M}(
    \hat{\rho}
    )-\mathcal{M}(
    \hat{\tau}
    )$ and $ds:=s_{\rho}
    -s_{\tau}$.
    Since $ds$ is finite and $dM$ is positive  according to Eq. (\ref{eq:thm:ent_cond}), we can choose a sufficiently small value of $\epsilon$ that always makes the difference positive. 

    The same conclusion holds for the other Pauli channel $\mathcal{E}_{\vec{\epsilon}}$ with small channel probabilities $|\vec{\epsilon}\;|\ll1$, given that the magnitude of $\vec{\epsilon}$ is finite.
    Therefore, we conclude that
    \begin{align}
    \label{eq:thm:small_ord_prv}
        \mathcal{M}\left[
        \Big(
        \mathcal{E}_{\vec{\epsilon}}
        \otimes\mathbb{I}
        \Big)
        \hat{\rho}
        \right]
        >
        \mathcal{M}\left[
        \Big(
        \mathcal{E}_{\vec{\epsilon}}
        \otimes\mathbb{I}
        \Big)
        \hat{\tau}
        \right].
    \end{align}
    
    The second part of the proof is to extend Eq. (\ref{eq:thm:small_ord_prv}) to the case of large probabilities.
    A Pauli channel can always be decomposed into a chain of equivalent Pauli channels as
    \begin{align}
        \mathcal{E}_{\vec{p}}
        =
        \underbrace{
        \mathcal{E}_{\vec{\epsilon}}\circ
        \mathcal{E}_{\vec{\epsilon}}\circ
        \ldots\circ\mathcal{E}_{\vec{\epsilon}}
        }_{L \mathrm{times}}
    \nonumber
    \end{align}
    with
    \begin{align}
        \epsilon_1&=\frac{1}{4}
        \left[
        1+q_{23}^{\frac{1}{L}}
        -q_{31}^{\frac{1}{L}}
        -q_{12}^{\frac{1}{L}}
        \right],
        \nonumber
        \\
        \epsilon_2&=\frac{1}{4}
        \left[
        1-q_{23}^{\frac{1}{L}}
        +q_{31}^{\frac{1}{L}}
        -q_{12}^{\frac{1}{L}}
        \right],
        \nonumber
        \\
        \epsilon_3&=\frac{1}{4}
        \left[
        1-q_{23}^{\frac{1}{L}}
        -q_{31}^{\frac{1}{L}}
        +q_{12}^{\frac{1}{L}}
        \right],
    \nonumber
    \end{align}
where $q_{ij}:=1-(p_i+p_j)/2$. By increasing $L$, we can choose $\vec{\epsilon}$ to be small enough to satisfy Eq. (\ref{eq:thm:small_ord_prv}).
Finally, we conclude that
    \begin{align}
        \mathcal{M}
        \Bigr[
        (\mathcal{E}_{\vec{p}}\otimes\mathbb{I})
        \hat{\rho}
        \Bigr]
        >
        \mathcal{M}
        \Bigr[
        (\mathcal{E}_{\vec{p}}\otimes\mathbb{I})
        \hat{\tau}
        \Bigr]
        \nonumber
    \end{align}
by iterating the first part of the proof $L$ times.
\end{proof}

Theorem \ref{th:ord_pre_pl} can be extended to encompass local Pauli channels, resulting in the following corollary:
    \begin{corollary}
    \label{cor:ord_pre_local_pl}
    Theorem \ref{th:ord_pre_pl} holds even when a local Pauli channel, denoted as $\mathcal{E}_{\vec{p}}\otimes\mathcal{E}_{\vec{p}'}$, is used instead of $\mathcal{E}_{\vec{p}}\otimes\mathbb{I}$.
    \end{corollary}
    \begin{proof}
    This can be easily demonstrated given that an arbitrary local Pauli channel can be decomposed as follows
    \begin{align}
        \mathcal{E}_{\vec{p}}
        \otimes\mathcal{E}_{\vec{p}'}
        =
        \Bigl(\mathcal{E}_{\vec{p}}
        \otimes\mathbb{I}\Bigr)
        \circ
        \Bigl(\mathbb{I}
        \otimes\mathcal{E}_{\vec{p}'}\Bigr).
    \nonumber
    \end{align}
    We end the proof by applying Theorem \ref{th:ord_pre_pl} to the channel $\mathcal{E}_{\vec{p}'}$ acting on the second qubit and subsequently to the channel $\mathcal{E}_{\vec{p}}$ through the first qubit. 
    \end{proof}
In the subsequent sections, we apply corollary \ref{cor:ord_pre_local_pl} to states affected by different Pauli noises, where the average probabilities of noises are same. It enhances the robustness of the numerical results regarding the entanglement boundaries influenced by the Pauli channel, as presented in section \ref{sec:sub:bnd_ent}. Furthermore, it establishes the relationship between these boundaries and those induced by noisy PBET, which is seen in section \ref{sec:sub:ent_of_pbt}.

\subsection{Upper and lower bounds of entanglement reduction}
\label{sec:sub:bnd_ent}

    \begin{figure}
        \centering
        \begin{subfigure}{.45\textwidth}
        \centering
        \includegraphics[width=\linewidth]
        {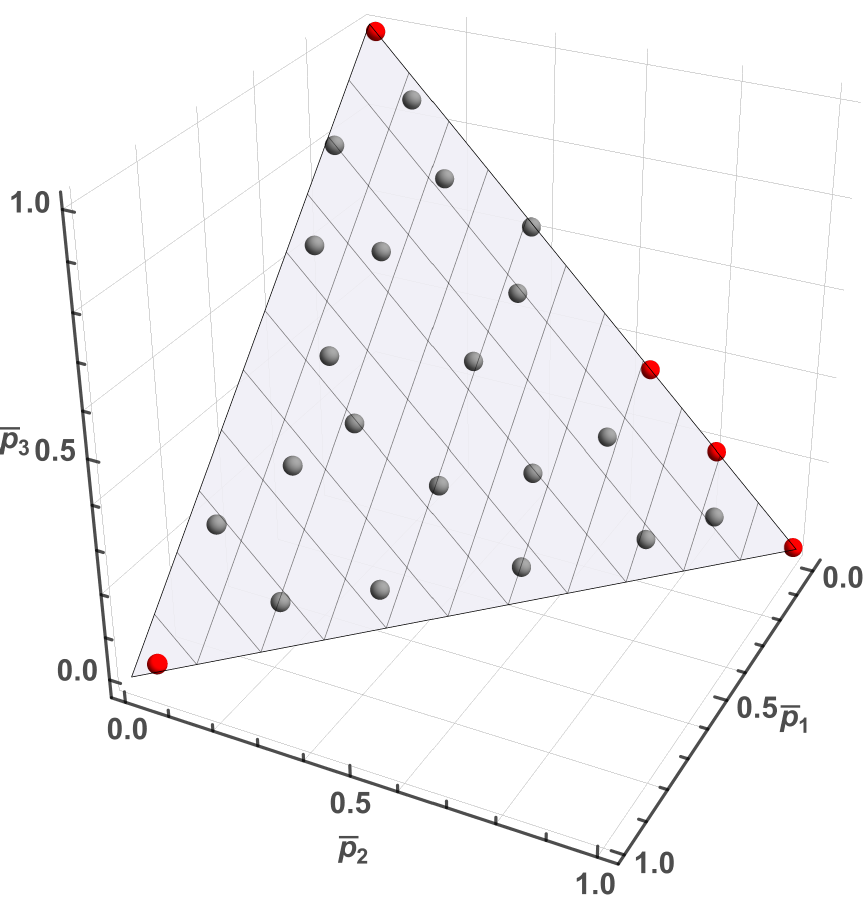}
        \caption{}
        \label{fig:PL_data_P}
        \end{subfigure}
    \hfill
        \begin{subfigure}{.4\textwidth}
        \centering
        \includegraphics[width=\linewidth]
        {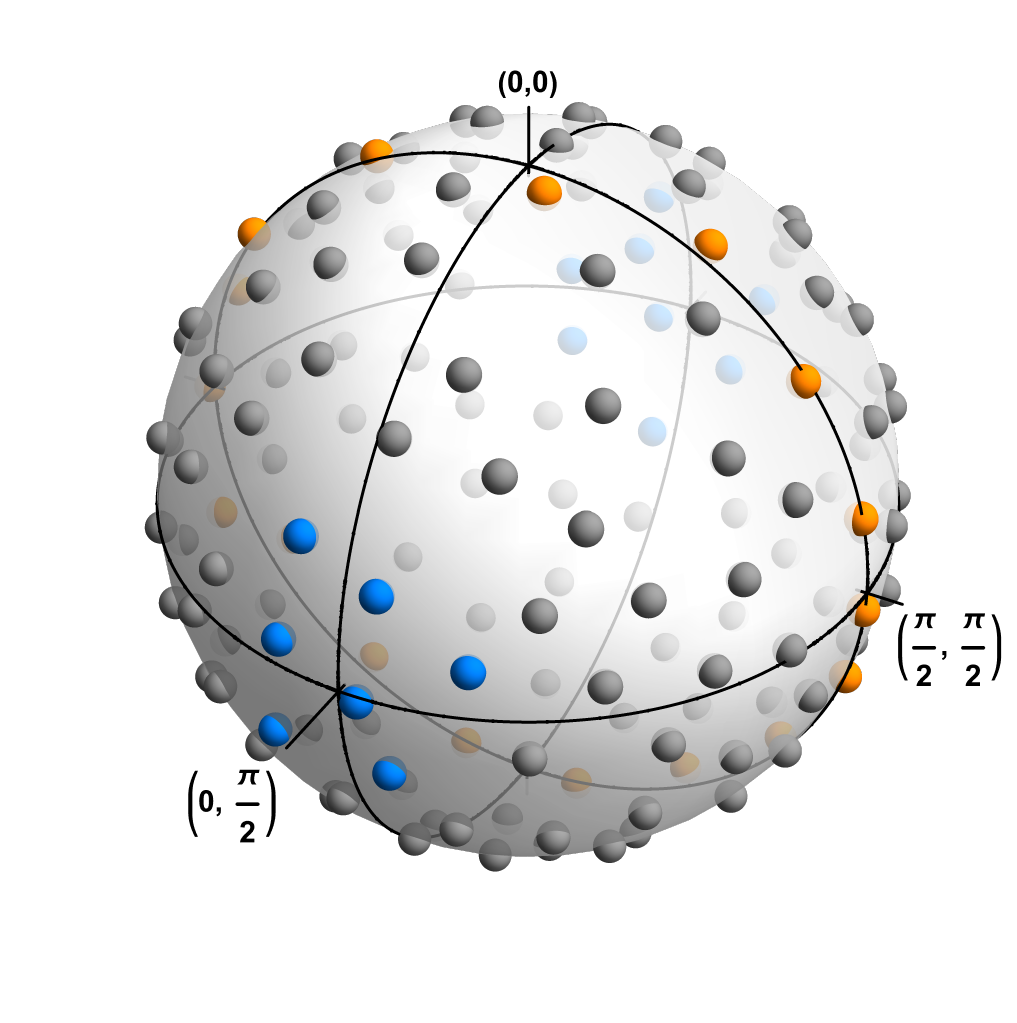}
        \caption{}
        \label{fig:PL_data_U}
        \end{subfigure}
        \caption{
        Sample sets of variables that represent arbitrary two-qubit state to find the boundaries of entanglement.
        (a) Randomly distributed sample set of relative channel probabilities $\bar{p}_i=p_{i}/\Omega\;(i=1,2,3)$, normalized to 1 with 25 elements as dots. It is represented on a plane with the sum of all terms equal to $1$.The red dots correspond to probabilities associated with the lowest $1000$ and highest $1000$ entanglements of possible states at the critical value.
        (b) Randomly distributed sample set of the first two Euler angles of unitary matrices for the first and second qubits, amounting to a total of $150$ elements as dots. It is represented over a shell with a radius of $1$. Blue dots represent angles associated with the lowest $1000$ entanglement values, while orange dots correspond to angles linked to the highest $1000$ entanglement values at the critical value.
        }
        \label{figs:PL_data}
    \end{figure}
    
Schmidt decomposition allows representing any two-qubit pure state as follows:
    \begin{align}
    \ket{\phi_z(\theta)}=\cos\left(\frac{\theta}{2}\right)
    \ket{0}\otimes\ket{0}
    +\sin\left(\frac{\theta}{2}\right)
    \ket{1}\otimes\ket{1}
    \label{eq:sts_pure_z}
    \end{align}
via local unitary transformation, where $0\le\theta\le\frac{\pi}{2}$. Therefore, we can express any two-qubit pure state as
    \begin{align}
        \ket{\phi}=\left[
            \hat{U}(\vec{\alpha})\otimes\hat{V}(\vec{\beta})
        \right]
        \ket{\phi_z(\theta)},
    \label{eq:sts_pure_arb}
    \end{align}
where $\hat{U}(\vec{\alpha})$ and $\hat{V}(\vec{\beta})$ represent Euler rotations with Euler angles $\vec{\alpha}:=(\alpha_1,\alpha_2,\gamma)$ and $\vec{\beta}:=(\beta_1,\beta_2,0)$, respectively.
We define the last term in $\vec{\beta}$ as zero because the z-axis rotation acting on the second qubit and the one acting on the first qubit yield the same result when applied to Eq. (\ref{eq:sts_pure_z}).
By exploiting Eq. (\ref{eq:measure_of_entanglement}), the measure of entanglement for state $\ket{\phi}$ is expressed as
    \begin{align}
    \label{eq:int_ent_single}
        M_0:=\mathcal{M}(\ket{\phi}\bra{\phi})=2\sin\theta,
    \end{align}
which is independent of the Euler angles.

For clarity, we assume that the two qubits are far apart but interact with the same environment, resulting in the same Pauli channel acting on both qubits separately.
According to this assumption, we define the entanglement of the state interacting with the environment as
    \begin{align}
        \mathcal{M}
        \Bigl[
            \left(
                \mathcal{E}_{\vec{p}}\otimes\mathcal{E}_{\vec{p}}
            \right)
            \ket{\phi}\bra{\phi}
        \Bigr].
    \end{align}

Fig. \ref{fig:PL_entanglement_3D} illustrates the upper and lower bounds of reduced entanglement through the Pauli channel concerning the initial entanglement $\mathcal{M}_0$ and average of channel probabilities $\Omega$.
Note that the condition for all possible states to exhibit non-zero entanglement occurs when the initial entanglement surpasses a critical value, denoted as $M_{\mathrm{low}}^c$, given by
    \begin{align}
        M_{\mathrm{low}}^c=
        -\frac{1}{2}
        \left(
        \left|1-\frac{3}{2}\Omega\right|
        -
        \left|1-\frac{3}{2}\Omega\right|^{-1}
        \right).
    \label{eq:cri_val}
    \end{align}
This critical value is depicted by the red line in Fig. \ref{fig:PL_entanglement_3D}. 

To determine the upper and lower bounds of the reduced entanglement precisely at this critical value, we generated a sample of $(\alpha_1,\alpha_2)$ and $(\beta_1,\beta_2)$ uniformly distributed over a shell with a radius of $1$. The sample is illustrated in Fig. \ref{fig:PL_data_U} utilizing polar coordinates. 
The angle $\gamma$ was sampled within the range of $0$ to $2\pi$ with intervals of $\pi/3$. Additionally, we used a sample with channel probabilities uniformly distributed over the plane $p_1 + p_2 + p_3 = 3\Omega$, as depicted in Fig. \ref{fig:PL_data_P}. We examined the measure of entanglement for all possible combinations of the samples we specified.

In Fig. \ref{fig:PL_data_P}, we highlight the channel probabilities of the lowest $1000$ and highest $1000$ entanglements for $\theta=0.2,0.4,0.6,0.8,1$ as red dots. From a more focused search within the data range of red dots, we discovered that the bit flip, bit-phase flip, and phase flip channels correspond to the channels of lower and upper bounds of entanglement.
    
Focusing on the phase flip channel $\mathcal{E}_p^3$, we examined $(\alpha_1,\alpha_2)$ and $(\beta_1,\beta_2)$ for the lowest $1000$ entanglements. These values are shown as blue dots in Fig. \ref{fig:PL_data_U}, while the values for the highest $1000$ entanglements are represented as orange dots.
From a more focused search, we found that the state corresponding to the lower bound is given by
    \begin{align}\label{eq:sts_low}
        \cos\left(\frac{\theta}{2}\right)\ket{+}\ket{+}
        +i\sin\left(\frac{\theta}{2}\right)\ket{-}\ket{-},
    \end{align}
whereas the state corresponding to the upper bound is given by
    \begin{align}\label{eq:sts_up}
        \cos\left(\frac{\theta}{2}\right)\ket{0}\ket{0}
        +\sin\left(\frac{\theta}{2}\right)\ket{1}\ket{1},
    \end{align}

where $\hat{\sigma}_1\ket{\pm}=\pm\ket{\pm}$.
Fig. \ref{fig:PL_data_U} shows that the orange dots form a band, indicating that Eq. (\ref{eq:sts_low}) and (\ref{eq:sts_up}) are not unique.
Given that the phase flip channel lacks $\hat{\sigma}_1$ and $\hat{\sigma}_2$ as Kraus operators, the measure of entanglement of its states remains invariant under a single-qubit z-axis rotation. 
Furthermore, by applying a unitary matrix to the states that transform the basis of $\hat{\sigma}_3$ to the basis of $\hat{\sigma}_1$ or $\hat{\sigma}_2$ , we can easily transform them into the boundary states of the bit or bit-phase flip channel, respectively.

Furthermore, we find that the lower and upper bounds obtained through the same numerical method in a region smaller than the critical value are entirely identical. We present the following theorem to support and extend the numerical results to the entire domain:
    \begin{theorem}
    \label{th:ent_single}
    Let $\mathcal{D}$ be a set of elements of two-qubit density matrices.
    Let $\mathcal{C}_\omega$ and $\mathcal{C}_\Omega$ be a set of single qubit Pauli channel with equivalent average of channel probabilities $\omega$ and $\Omega$ smaller than $2/3$, respectively.
    Let us define that
    \begin{align}\label{eq:ent_assum_min}
        &\mathcal{M}
        \Bigl[
        (\mathcal{E}_{\vec{p}_\mathrm{m}}
        \otimes\mathbb{I})
        \hat{\rho}_\mathrm{m}
        \Bigr]:=
        \min_{\hat{\rho}\in\mathcal{D},\mathcal{E}\in\mathcal{C}_\Omega}
        \mathcal{M}
        \Bigl[
        (\mathcal{E}
        \otimes\mathbb{I})
        \hat{\rho}
        \Bigr]>0,
        \\
        &\mathcal{M}
        \Bigl[
        (\mathcal{E}_{\vec{p}_\mathrm{M}}
        \otimes\mathbb{I})
        \hat{\rho}_\mathrm{M}
        \Bigr]:=
        \max_{\hat{\rho}\in\mathcal{D},\mathcal{E}\in\mathcal{C}_\Omega}
        \mathcal{M}
        \Bigl[
        (\mathcal{E}
        \otimes\mathbb{I})
        \hat{\rho}
        \Bigr].
    \end{align}    
    Then for all the averages $\omega\in(0,\Omega)$,
    \begin{align}
        \mathcal{M}
        \Bigl[
        (\mathcal{E}_{\vec{q}_\mathrm{m}}
        \otimes\mathbb{I})
        \hat{\rho}_\mathrm{m}
        \Bigr]&=
        \min_{\hat{\rho}\in\mathcal{D},\mathcal{E}\in\mathcal{C}_\omega}
        \mathcal{M}
        \Bigl[
        (\mathcal{E}
        \otimes\mathbb{I})
        \hat{\rho}
        \Bigr],\\
        \mathcal{M}
        \Bigl[
        (\mathcal{E}_{\vec{q}_\mathrm{M}}
        \otimes\mathbb{I})
        \hat{\rho}_\mathrm{M}
        \Bigr]&=
        \max_{\hat{\rho}\in\mathcal{D},\mathcal{E}\in\mathcal{C}_\omega}
        \mathcal{M}
        \Bigl[
        (\mathcal{E}
        \otimes\mathbb{I})
        \hat{\rho}
        \Bigr],
    \end{align}
    where $\vec{q}_\mathrm{m}:=(\omega/\Omega)\vec{p}_\mathrm{m}$ and $\vec{q}_\mathrm{M}:=(\omega/\Omega)\vec{p}_\mathrm{M}$.
    \end{theorem}
    \begin{proof}
    We only provide proof for the minimum case, which is analogous to the maximum case.
    Let us assume that there exist a channel $\mathcal{E}_{\vec{q}}\in\mathcal{C}_{\omega}$ and a density matrix $\hat{\rho}\in\mathcal{D}$ that satisfy
    \begin{align}
        \mathcal{M}
        \Bigl[
        \Bigl(\mathcal{E}_{\vec{q}}
        \otimes\mathbb{I}\Bigr)
        \hat{\rho}
        \Bigr]
        <
        \mathcal{M}
        \Bigl[
        \Bigl(\mathcal{E}_{\vec{q}_\mathrm{m}}
        \otimes\mathbb{I}\Bigr)
        \hat{\rho}_\mathrm{m}
        \Bigr].
    \nonumber
    \end{align}
    The channel probabilities of channel $\mathcal{E}_{\vec{r}}$ satisfying $\mathcal{E}_{\vec{r}}\circ\mathcal{E}_{\vec{q}_\mathrm{m}}=\mathcal{E}_{\vec{p}_\mathrm{m}}$ are
    \begin{align}
    \vec{r}=4\mathbf{Q}^{-1}(\vec{p}_\mathrm{m}-\vec{q}_\mathrm{m}),
    \nonumber
    \end{align}
    where
    \begin{align}
        \mathbf{Q}:=
        q_0\mathbf{I}-
        \left[\begin{array}{ccc}
            q_1    &q_1-q_3   &q_1-q_2       \\
            q_2-q_3&q_2       &-q_1+q_2        \\
            -q_2+q_3 &-q_1+q_3    &q_3
        \end{array}\right]
    \nonumber
    \end{align}
    with $\vec{q}_\mathrm{m}:=(q_1,q_2,q_3)$, $q_0:=4-(q_1+q_2+q_3)$ and $\mathbf{I}$ is the identity matrix.
    However, the matrix $\mathbf{Q}$ cannot be inverted when
    \begin{align}
    (q_1+q_2-2)(q_2+q_3-2)(q_3+q_1-2)=0.
    \nonumber
    \end{align}
    Given that we are considering $3\omega=q_1+q_2+q_3<2$ regime, the matrix is always invertible.
    As a result of Theorem \ref{th:ord_pre_pl}, we can derive
    \begin{align}
        \mathcal{M}
        \Bigl[\Bigl((
        \mathcal{E}_{\vec{r}}
        \circ
        \mathcal{E}_{\vec{q}})
        \otimes
        \mathbb{I}\Bigr)
        \hat{\rho}
        \Bigr]
        &<
        \mathcal{M}
        \Bigl[\Bigl((
        \mathcal{E}_{\vec{r}}
        \circ
        \mathcal{E}_{\vec{q}_\mathrm{m}})
        \otimes
        \mathbb{I}\Bigr)
        \hat{\rho}_\mathrm{m}
        \Bigr],
        \nonumber
    \end{align}
which becomes
    \begin{align}
        \mathcal{M}
        \Bigl[\Bigl((
        \mathcal{E}_{\vec{r}}
        \circ
        \mathcal{E}_{\vec{q}})
        \otimes
        \mathbb{I}\Bigr)
        \hat{\rho}
        \Bigr]
        &<
        \mathcal{M}
        \Bigl[
        \Bigl(\mathcal{E}_{\vec{p}_\mathrm{m}}
        \otimes\mathbb{I}\Bigr)
        \hat{\rho}_\mathrm{m}
        \Bigr].
        \nonumber
    \end{align}
    Given that the inequality we  derived restrict Eq. (\ref{eq:ent_assum_min}),
    we can conclude that $\mathcal{M}
        \Bigl[
        \Bigl(\mathcal{E}_{\vec{q}_\mathrm{m}}
        \otimes\mathbb{I}\Bigr)
        \hat{\rho}_\mathrm{m}
        \Bigr]$ is also the minimum. 
    \end{proof}
According to Theorem \ref{th:ent_single}, we obtain the following corollary:
    \begin{corollary}
    \label{cor:bnd_ent}
        Theorem \ref{th:ent_single} can be extended to local Pauli channel.
    \end{corollary}
    \begin{proof}
        The proof follows from Theorem \ref{th:ent_single} following the same approach as Corollary \ref{cor:ord_pre_local_pl}.
    \end{proof}
    
Corollary \ref{cor:bnd_ent} demonstrates that the states and corresponding channels of boundary entanglement obtained at the critical value remain consistent within the regime where the average is smaller than that of the critical value.

    \begin{figure}
        \centering
        \begin{subfigure}{.45\textwidth}
        \centering
        \includegraphics[width=\linewidth]
        {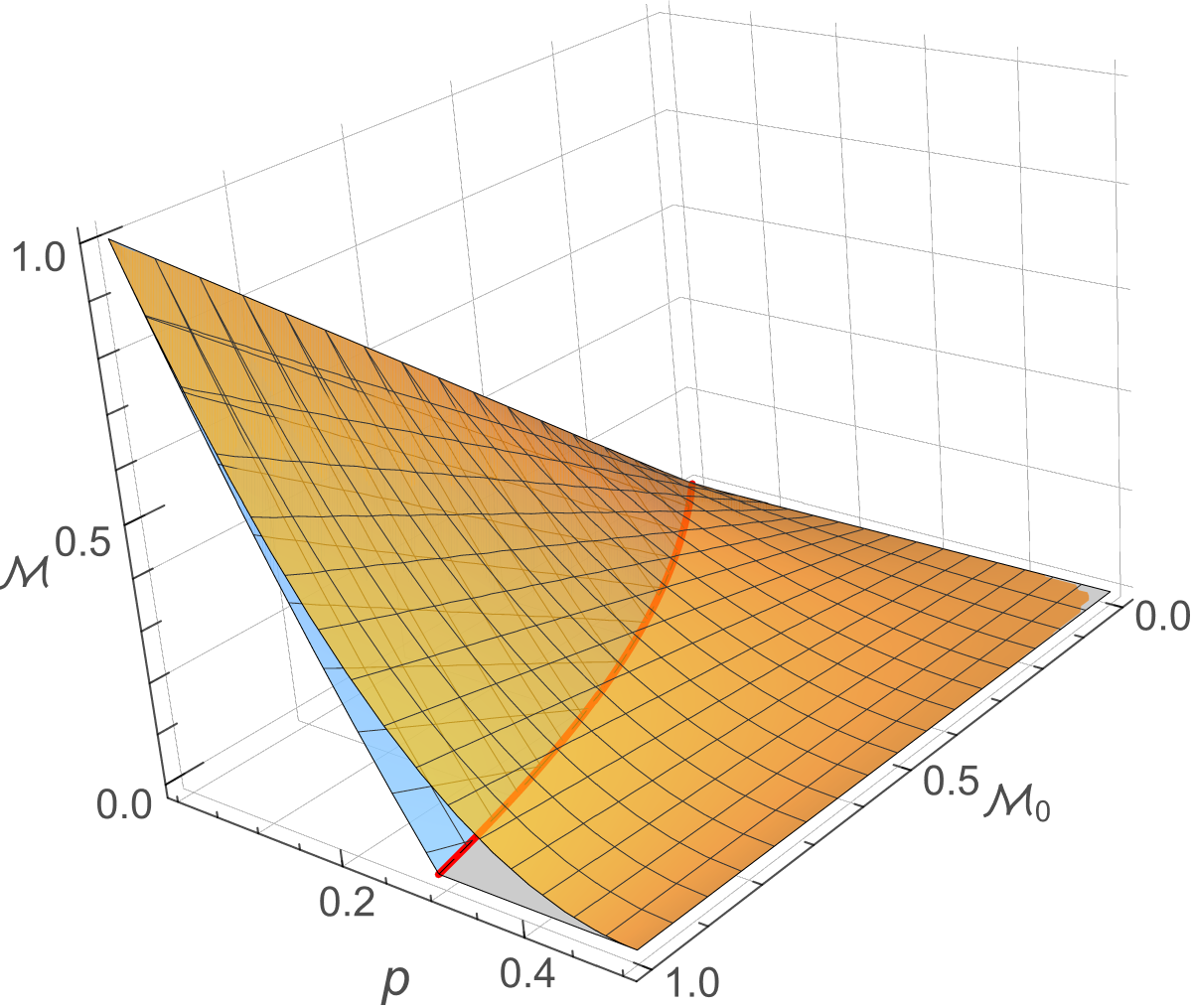}
        \caption{}
        \label{fig:PL_entanglement_3D}
        \end{subfigure}
    \hfill
        \begin{subfigure}{.4\textwidth}
        \centering
        \includegraphics[width=\linewidth]
        {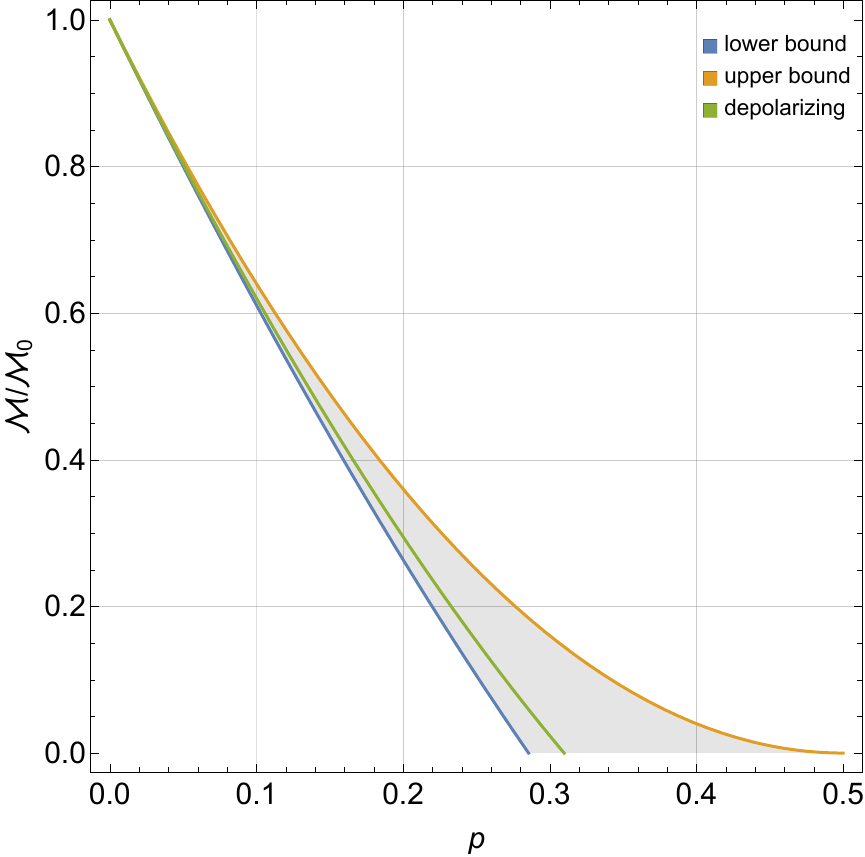}
        \caption{}
        \label{fig:PL_entanglement_2D}
        \end{subfigure}
        \caption{
        (a) Measure of entanglement after passing through a local Pauli channel with respect to entanglement $M_0$ for an initial pure state and average of probabilities for a single qubit Pauli channel. The blue surface corresponds to the lower bound whereas the orange surface corresponds to the upper bound. The red line corresponds to the critical value which represents the boundary where the entanglement of all states is larger than zero.
        (b) Entanglement divided by the initial entanglement against
        the average of the probabilities at $M_0 = 0.8$. Blue, orange, and green lines correspond to entanglement of the lower bound, upper bound, and depolarizing channel, respectively. The gray region corresponds to the possible entanglement that the arbitrary states may have. 
        }
        \label{figs:PL_entanglemet}
    \end{figure}

Fig. \ref{fig:PL_entanglement_3D} presents the upper and lower bounds of reduced entanglement in relation to the entanglement of the initial state and average of channel probabilities.
That the entanglement for the lower bound is $M_{\mathrm{low}}$ given by
    \begin{align}\label{eq:ent_pl_low}
        M_{\mathrm{low}}=\mathrm{max}
        \left[
        0,
        M_0\left|1-\frac{3}{2}\Omega\right|+\frac{1}{2}
        \left(1-\frac{3}{2}\Omega\right)^2
        -\frac{1}{2}
        \right],
    \end{align}
and for the upper bound $M_{\mathrm{up}}$ is expressed as
    \begin{align}\label{eq:ent_pl_up}
        M_{\mathrm{up}}=\mathrm{max}
        \left[
        0,M_0\left(1-\frac{3}{2}\Omega\right)^2
        \right].
    \end{align}
These expressions correspond to the entanglements derived from Eq. (\ref{eq:sts_low}) and (\ref{eq:sts_up}) through the phase flip channel, respectively.
We additionally confirmed, employing the same search method, that the lower and upper bounds remain consistent for entanglements larger than the critical value. 

For depolarized channels, it is worth noting that all states with the same initial entanglement will yield the same entanglement measure as each other after passing through the channel.
The reduced entanglement for this channel is calculated from \cite{lee2000entanglement}, which can be expressed using our notation as
    \begin{align}
        M_{\mathrm{dep}}=
        \max\left[0,
        M_0(1-\Omega)^2-\Omega
        \left(
        1-\frac{1}{2}\Omega
        \right)
        \right],
    \end{align}
where Werner state can be represented as a state maximally entangled passing through depolarizing channel.

Fig. \ref{fig:PL_entanglement_2D} shows the relative reduced entanglement plotted against the average of channel probabilities at $\mathcal{M}_0=0.8$.
Notably, the entanglement of the depolarizing channel is near to the lower bound than to the upper bound.

\section{Noisy Port-based Teleportation}
\label{sec:noisy_pbt}
\subsection{Channel description}
\label{sec:sub:ch_descpt}
We first review the ideal case where the resource states are not disturbed by the environment.
Even if maximally entangled states are prepared as a resource state, the unknown state sent to the receiver is deformed owing to the finite number of ports.
It was shown in \cite{ishizaka2009quantum} that the ideal scheme can be described isng a depolarizing channel as
    \begin{align}\label{eq:_PBT_channel_0}
        \Lambda(\hat{\rho})
        =\left(1-q_N\right)
        \frac{\hat{\sigma}_0}{2}
        +
        q_N
        \hat{\rho},
    \end{align}
where $q_N:=2f(\Lambda)-1$ and $f(\Lambda)$ is defined by Eq. (\ref{eq:rel_tel_fid_and_ent_fid}) and (\ref{eq:ent_fid_PBT}).
This result follows from the fact that both the resource states and the measurement operators can be represented as Bell states, as shown in Eq. (\ref{eq:resource_sts}) and (\ref{eq:povm_meas}), given that the Bell state have the property of being invariant to twirling and can be expressed as
    \begin{align}
        \ket{\Psi^0}=
        \left(\hat{U}^*\otimes\hat{U}\right)\ket{\Psi^0},
    \end{align}
where $\hat{U}$ is an arbitrary unitary matrix.

Hereafter, we assume that there is local Pauli noise due to the environment.
Then the noisy resource state is given by
    \begin{align}
    \label{eq:sts_res_pl}
        &\left(
        \mathcal{E}_{\vec{p}}\otimes\mathcal{E}_{\vec{p}}
        \right)
        \ket{\Psi^0}\bra{\Psi^0}
        =\sum^3_{k=0}\frac{\alpha_k}{16}
        \ket{\Psi^k}\bra{\Psi^k}
    \end{align}
with 
    \begin{align}
        \alpha_0:=p_0^2+p_1^2+p_2^2+p_3^2\;&,\;
        \alpha_1:=2
        \left(
        p_0p_1+p_2p_3
        \right),
        \nonumber
        \\
        \alpha_2:=2
        \left(
        p_0p_2+p_3p_1
        \right)\;&,\;
        \alpha_3:=2
        \left(
        p_0p_3+p_1p_2
        \right),
        \nonumber
    \end{align}
and the four possible Bell states are
    \begin{align}
        \ket{\Psi^1}:=
        \frac{\ket{0}\otimes\ket{1}+\ket{1}\otimes\ket{0}}
        {\sqrt{2}},
        \nonumber
        \\
        \ket{\Psi^2}:=
        \frac{\ket{0}\otimes\ket{1}-\ket{1}\otimes\ket{0}}
        {\sqrt{2}i},
        \nonumber
        \\
        \ket{\Psi^3}:=
        \frac{\ket{0}\otimes\ket{0}-\ket{1}\otimes\ket{1}}
        {\sqrt{2}},
        \nonumber
    \end{align}
where the first Bell state is defined by Eq. (\ref{eq:bell_sts}).
Given that Bell states are stabilizer states of two qubits, they can be transformed into each other using local Pauli matrices and local Clifford matrices \cite{nest2004efficient}.
In particular, Bell states $i=1,2,3$ can be represented by applying Pauli matrices on the first qubit of the first Bell state as follows:
    \begin{align}\label{eq:relation_bell}
        \ket{\Psi^i}=
        \left(
        \hat{\sigma}_i\otimes\hat{\sigma}_0
        \right)
        \ket{\Psi^0}.
    \end{align}

Substituting the noisy resource state instead of the noiseless one at the channel defined as Eq. (\ref{eq:PBT_channel}), 
it can be shown that the channel the noisy PBT is
    \begin{align}
        \Lambda_{\vec{p}}\left(\hat{\rho}\right)
        &=
        \sum_{k=0}^3
        \frac{\alpha_k}{16}\hat{\sigma}_k
        \Lambda\left(\hat{\rho}\right)
        \hat{\sigma}_k
        \nonumber\\
        &=(1-q_N)\frac{\hat{\sigma}_0}{2}
        +q_N\sum^3_{k=0}\frac{\alpha_k}{16}
        \hat{\sigma}_k\hat{\rho}\hat{\sigma}_k
        \nonumber\\
        \label{eq:ch_PBT_pl}
        &=\frac{1+3q_N q_{\vec{p}}}{4}
        \hat{\rho}+\sum^3_{j=1}\frac{1-q_Nq^{(j)}_{\vec{p}}}{4}
        \hat{\sigma}_j\hat{\rho}\hat{\sigma}_j,
        \end{align}
where we used Eq. (\ref{eq:sts_res_pl}) and (\ref{eq:relation_bell}) at the first expression. The last expression is represented by
    \begin{align}
        \label{eq:PBT_q_0}
        q_{\vec{p}}&:=\frac{1}{3}\sum_{j=1}^3(1-p_j)^2
        -\frac{1}{12}\left[
        \Delta^2_{12}+\Delta^2_{23}+\Delta^2_{31}
        \right],\\
        q^{(1)}_{\vec{p}}&:=(1-p_1)^2-\frac{1}{4}\left[
        \Delta^2_{12}-\Delta^2_{23}+\Delta^2_{31}
        \right],\\
        q^{(2)}_{\vec{p}}&:=(1-p_2)^2-\frac{1}{4}\left[
        \Delta^2_{12}+\Delta^2_{23}-\Delta^2_{31}
        \right],\\
        q^{(3)}_{\vec{p}}&:=(1-p_3)^2-\frac{1}{4}\left[
        -\Delta^2_{12}+\Delta^2_{23}+\Delta^2_{31}
        \right],
    \end{align}
where $\Delta_{ij}:=|p_i-p_j|$. Note that 
$q_{\vec{p}}$ is the average of channel probabilities $\vec{q}_{\vec{p}}:=\Bigl(q^{(1)}_{\vec{p}},q^{(2)}_{\vec{p}},q^{(3)}_{\vec{p}}\Bigr)$.
From Eq. (\ref{eq:ch_PBT_pl}), we obtain the teleportation fidelity as
    \begin{align}\label{eq:tel_fid_PBT_pl}
        f(\Lambda_{\vec{p}})
        =\frac{1}{2}+\frac{1}{2}q_Nq_{\vec{p}},
    \end{align}
and the entanglement fidelity as
    \begin{align}\label{eq:ent_fid_PBT_pl}
        F(\Lambda_{\vec{p}})
        =\frac{1}{4}+\frac{3}{4}q_Nq_{\vec{p}}.
    \end{align}
It is important to note that both the teleportation fidelity and entanglement fidelity converge to those of standard PBT under noiseless conditions. Furthermore, at the asymptotic limit of $N\rightarrow\infty$, both fidelities approach those of original teleportation using the noisy resource state we used.

The  channel of noisy PBT can be decomposed into a chain of two channels as
    \begin{align}\label{eq:ch_PBT_chain}
        \Lambda_{\vec{p}}(\hat{\rho})
        =
        \left[
        \mathcal{E}^{\mathrm{dep}}_{1-q_N}
        \circ
        \mathcal{E}_{\vec{1}-\vec{q}_{\vec{p}}}
        \right]
        \hat{\rho},
    \end{align}
where $\vec{1}:=(1,1,1)$.
In Eq. (\ref{eq:ch_PBT_chain}), $\mathcal{E}^{\mathrm{dep}}_{1-q_N}$ represents the depolarization due to the limitation in the number of ports, and $\mathcal{E}_{\vec{1}-\vec{q}_{\vec{p}}}$ represents the noise introduced by the noisy resource state.
Given that the depolarizing channel commutes with any other Pauli channel, there is no order between $\mathcal{E}^{\mathrm{dep}}_{1-q_N}$ and $\mathcal{E}_{\vec{1}-\vec{q}_{\vec{p}}}$. Therefore, both losses resulting from different causes can be considered simultaneous. 

In particular, for the depolarizing form of noise at the resource state, the environment noise channel is represented as a depolarizing channel with probability $q_p:=(1-p)^2$ and the entanglement fidelity is represented as
    \begin{align}
        F(\Lambda_{\vec{p}})
        =\frac{1}{4}+\frac{3}{4}q_Nq_{p}.
    \end{align} 
\subsection{Measure of entanglement}
\label{sec:sub:ent_of_pbt}
Let us consider port-based entanglement teleportation (PBET), the scheme to teleport an unknown two-qubit pure state $\ket{\phi}$ with noisy PBT for each qubit.
If there is no noise,
the measure of entanglement $M^{\mathrm{free}}$ after teleportation is given as
    \begin{align}
        M^{\mathrm{free}}=\mathrm{max}\left[
        0,
        -\frac{1}{2}
        +q_N^2
        \left(M_0+\frac{1}{2}\right)
        \right],
    \label{eq:ent_PBT_ideal}
    \end{align}
where $M_0$ is the measure of entanglement of the prepared unknown state defined as Eq. (\ref{eq:int_ent_single}). We can see that entanglement teleportation works perfectly in the asymptotic limit.

If we consider local depolarizing noise on the resource state in Eq. (\ref{eq:PBT_channel}), the entanglement after teleportation $M^{\mathrm{dep}}_p$ is given by
    \begin{align}
        M^{\mathrm{dep}}_p=\mathrm{max}\left[
        0,-\frac{1}{2}
        +q_N^2
        \left(\frac{1+2N_{\mathrm{dep}}}{3}\right)^2
        \left(M_0+\frac{1}{2}\right)
        \right],\label{eq:ent_PBT_dep}
    \end{align}
where the entanglement of the noisy resource state is expressed as 
    \begin{align}
        N_{\mathrm{dep}}&=
        \mathcal{M}
        \Bigl[
            \left(\mathcal{E}^{\mathrm{dep}}_p
            \otimes\mathcal{E}^{\mathrm{dep}}_p\right)
            \ket{\Psi^0}\bra{\Psi^0}
        \Bigr]
        \nonumber\\
        &=\mathrm{max}\left[0,\frac{3q_p-1}{2}\right].
    \end{align}

We extend the noise of the resource state to Pauli noise.
It is not guaranteed that a phase flip channel exists in a set of constant average of probabilities $\vec{1}-\vec{q}_{\vec{p}}$ with respect to physical probabilities $\vec{p}$.
There are regions where the probabilities of noise $\vec{p}$ do not allow for a channel of environment noise $\mathcal{E}_{\vec{1}-\vec{q}_{\vec{p}}}$ to have a phase flip shape.
The Lagrange multiplier method was applied to determine the range of the average containing the phase flip channel.
Given that $1-q_{\vec{p}}^{(1)}$ is a channel probability, it is not larger than 1.
If the bit flip channel is a possible form, the maximum probability for variables $\vec{p}$
with constraint
    \begin{align}
        q_{\vec{p}}^{(1)}=q_{\vec{p}}^{(2)},
    \nonumber
    \end{align}
is 1.
With a long but simple calculation, it is possible to find the range with the maximum equal to 1:
    \begin{align}
        \frac{2}{3}\leq q_{\vec{p}}\leq 1.
    \end{align}
    
By exploiting Corollary \ref{cor:bnd_ent}, the initial states and corresponding channel of the boundaries presented in section \ref{sec:sub:bnd_ent} are also the boundaries of the teleportation.
As a result, the boundaries of the entanglement are represented as
    \begin{align}
        M_{\mathrm{bound}}^{\mathrm{PBT}}&=\mathcal{M}\left[
        \Bigl(\tilde{\mathcal{E}}^{\mathrm{dep}}_{1-q_N}\circ
        \tilde{\mathcal{E}}^3_{3-3q_{\vec{p}}}\Bigr)
        \hat{\Psi}_\mathrm{bound}
        \right]
        \nonumber\\
        &=
        \mathrm{max}
        \Bigl[0,M^\mathrm{env}_{\mathrm{bound}}
        q_N^2-\frac{1-q_N^2}{2}
        \Bigr],
    \end{align}
where $\mathrm{bound}$ is $\mathrm{up}$ and $\mathrm{low}$ for the upper and lower bounds, respectively.
In the first equation, we made use of the notation 
$\hat{\Psi}_\alpha:=\ket{\psi_\alpha}\bra{\psi_\alpha}$ and 
$\tilde{\mathcal{E}}_\alpha^\beta:=(\mathcal{E}_\alpha^\beta\otimes\mathcal{E}_\alpha^\beta)$
for simplicity. The second equation is expressed in term of $q_N$ and $q_{\vec{P}}$, where $M^\mathrm{env}_{\mathrm{up}}$ and $M^\mathrm{env}_{\mathrm{low}}$ are the substitution of $(1-q_{\vec{p}})$ for $\Omega$ in Eq. (\ref{eq:ent_pl_low}) and (\ref{eq:ent_pl_up}), respectively.
 
At the asymptotic limit of $N\rightarrow\infty$ and small amount of error $\Omega=(p_1+p_2+p_3)/3\ll 1$, the lower and upper bounds are approximated by
    \begin{align}
        M^{\mathrm{PBT}}_{\mathrm{low}}
        &\rightarrow 
        \mathrm{max}\Bigl[
        0,
        M_0
        -6\Omega\frac{M_0+1}{2}
        -\frac{1}{N}\left(2 M_0+1\right)
        \Bigr],\\
        M^{\mathrm{PBT}}_{\mathrm{up}}
        &\rightarrow
        \mathrm{max}\Bigl[
        0,
        M_0
        -6\Omega M_0
        -\frac{1}{N}\left(2 M_0+1\right)
        \Bigr],
    \end{align}
where we have approximated Eq. (\ref{eq:PBT_q_0}) as
    \begin{align}
        q_{\vec{p}}\rightarrow1-2\Omega.
    \end{align}
For both equations, the first term represents the initial entanglement, the second term denotes the loss due to environment noise, and the last term accounts for the loss attributed to the limit on the number of ports.
It is evident that the second term is proportional to the average of noise probabilities $\Omega$, and the last term is reciprocal to the number of ports $N$.

Let us focus on the second term first.
When the unknown state initially possesses maximum entanglement, $M_0 = 1$, the loss exactly matches 6 times the average of noise probabilities, which is $6\Omega$. 
As the measure of entanglement decreases, the second term in the upper bound, which represents the minimum loss, consistently remains equal to this value, and the upper and lower boundaries widen to $(1-M_0)/2$.
Turning our attention to the last term, we can see that its slope ranges from 0 to 3, given that $M_0$ varies between 0 to 1.
In addition, the slope remains constant at $2M_0+1$
for states with the same initial entanglement, as the teleported entanglement of lower and upper bounds are equivalent.

\section{Conclusions}
\label{sec:con}
In this study, we conducted an in-depth exploration of entanglement teleportation through noisy PBT. We determined the boundaries of the measure of entanglement for teleported unknown states as a function of initial entanglement and the average of channel probabilities. Specifically, we delved into the behavior of entanglement loss in the asymptotic limit of the number of ports and at the existence of small amount of noise.

Our findings reveal that the loss of entanglement due to the limited number of ports is less than the inverse of the number of ports, unaffected by noise, and larger for stronger entanglement of the unknown state. The other loss due to environmental noise is proportional to the average of channel probabilities, with the maximum slope equal to 6 times of the measure of entanglement; smaller entanglement can result in smaller slope depending on the unknown state. In the course of deriving our results, we proved that the order of entanglement of two-qubit states is preserved under the influence of the local Pauli channel. Moreover, we determined the entanglement boundaries of the channel.

The PBET protocol offers a significant advantage in implementing quantum entanglement distribution \cite{chou2007functional} in practice, as it eliminates the need for quantum correction by the receiver. 
Through an analysis of the standard PBET, we have demonstrated that entanglement can be efficiently distributed by a protocol utilizing resource states of large size in the presence of noise.
We anticipate that the determination of the boundaries of teleported entanglement will provide a reliable guide for developing optimized PBET protocols.

There are several intriguing questions that remain unanswered concerning noisy PBT and beyond.
For instance, we could explore how entanglement loss varies when the resource state is subjected to amplitude damping noise.
Furthermore, we can investigate entanglement teleportation within different variants of PBT protocols or even more general teleportation schemes, extending our research beyond the scope of the standard PBT studied here. These future research avenues hold great promise and can build upon the foundational knowledge acquired through the study of noisy PBT.

\section*{Acknowledgements}

This research was supported by Creation of the Quantum Information Science R\&D Ecosystem through the National Research Foundation of Korea funded by Ministry of Science (Grant No. NRF-2023R1A2C1005588). K.J. acknowledges support by the National Research Foundation of Korea through a grant funded by the Ministry of Science and ICT (NRF-2022M3H3A1098237), the Ministry of Education (NRF-2021R1I1A1A01042199), and Korea Institute of Science and Technology Information (P23031).
All numerical calculations and figures were performed using Wolfram Research, Inc., Mathematica, Version 13.3, Champaign, IL (2023).

%

\end{document}